\documentclass[10pt,conference]{IEEEtran}

\usepackage{research5}
\usepackage{psfrag}
\usepackage{epsfig}
\usepackage{graphicx}
\usepackage{multicol}

\allowdisplaybreaks[3]

\newtheorem{thm}{Theorem}
\newtheorem{lem}{Lemma}

\newtheorem{dfn}{Definition}

\newcommand{\smooth}{\delta}
\newcommand{\err}{\epsilon}

\newcommand{\dd}{\textnormal{d}}

\begin{document}

\title{Simple Channel Coding Bounds}

\author{\authorblockN{Ligong Wang}
\authorblockA{Signal and Information Processing Laboratory\\
ETH Zurich, Switzerland\\
\texttt{ wang@isi.ee.ethz.ch}}\and
\authorblockN{Roger Colbeck}
\authorblockA{Institute for Theoretical Physics, and\\
Institute of Theoretical Computer Science\\
ETH Zurich, Switzerland\\
\texttt{colbeck@phys.ethz.ch}}\and
\authorblockN{Renato Renner}
\authorblockA{Institute for Theoretical Physics\\
ETH Zurich, Switzerland\\
\texttt{renner@phys.ethz.ch}}}

\maketitle

\begin{abstract}
  New channel coding converse and achievability bounds are
  derived for a single use of an arbitrary channel. Both bounds are
  expressed using a quantity called the ``smooth 
  $0$-divergence'', which is a generalization of R\'enyi's divergence
  of order~$0$. The bounds are also studied in the limit of large
  block-lengths. In particular, they combine 
  to give a general capacity formula which is equivalent to the one
  derived by 
  Verd\'u and Han.
\end{abstract}

\section{Introduction}



We consider the problem of transmitting information through a channel.
A channel consists of an input alphabet $\set{X}$, an output alphabet
$\set{Y}$, where $\set{X}$ and $\set{Y}$ are each equipped with a
$\sigma$-Algebra, and the channel law which is a stochastic kernel
$P_{Y|X}$ from $\set{X}$ to
$\set{Y}$.  We consider average error probabilities throughout this
paper,\footnote{Note that Shannon's method of obtaining codes that
  have small  
  maximum error probabilities from those that have small average error
  probabilities \cite{shannon48} can be applied to our codes. We shall
  not examine other such methods which might lead to tighter bounds for
  finite block-lengths.}  thus an
\emph{$(m,\err)$-code} consists of an encoder $f: 
\{1,\ldots,m\}\to\set{X}, i\mapsto x$ and a decoder
$g:\set{Y}\to\{1,\ldots,m\}, y \mapsto \hat{i}$ such that the
probability that
$i\neq \hat{i}$ is smaller than or equal 
to $\err$, assuming that the message is uniformly
distributed. Our aim is to derive upper and lower bounds
on the largest $m$ given $\err>0$ such that an $(m,\err)$-code
exists for a given channel.


Such bounds are different from those in Shannon's original work
\cite{shannon48} in the sense that they are
nonasymptotic and do not rely on any channel structure such as
memorylessness or information stability. 

Previous works have demonstrated the advantages of such nonasymptotic
bounds. They can lead to more general channel capacity formulas
\cite{verduhan94} as well as giving tight approximations to the maximal
rate achievable for a desired error probability and a fixed block-length
\cite{polyanskiypoorverdu08}. 

In this paper we prove a new converse bound and a new achievability
bound. They are asymptotically tight in the sense that they combine to
give a general capacity formula that is equivalent to
\cite[(1.4)]{verduhan94}. We are mainly interested in
proving simple bounds which offer theoretical intuitions into channel
coding problems. It is not our main concern to derive bounds which
outperform the existing ones in estimating the largest achievable
rates in finite block-length scenarios. In fact, as will be seen in
Section~\ref{sec:compare}, the new achievability bound is less tight than
the one in~\cite{polyanskiypoorverdu08}, though the differences are
small. 

Both new bounds are expressed using a quantity which we call the \emph{smooth
  $0$-divergence}, denoted as $D_0^\smooth(\cdot\| \cdot)$ where $\smooth$ is a
positive parameter. This quantity is a generalization of R\'enyi's
divergence of 
order $0$ \cite{renyi61}. Thus, our new bounds demonstrate connections
between the channel coding problem and R\'enyi's divergence of
order~$0$. Various previous works
\cite{gallager65,arimoto77,csiszar95} have shown
connections between channel coding and R\'enyi's information
measures of order~$\alpha$ for $\alpha\ge \frac{1}{2}$. Also relevant
is  \cite{rennerwolfwullschleger06} where channel coding 
bounds were derived using the smooth min-
and max-entropies introduced in~\cite{rennerwolf04}.


As will be seen, proofs of the new bounds are
simple and self-contained. The achievability bound uses random coding
and suboptimal decoding, where the decoding rule can be thought of as
a generalization of Shannon's joint typicality decoding rule
\cite{shannon48}. The converse is proved by simple algebra
combined with the fact that $D_0^\smooth (\cdot\| \cdot)$ satisfies a
Data Processing Theorem. 

The quantity $D_0^\smooth (\cdot\| \cdot)$ has also been defined for
quantum systems 
\cite{datta08,WLGrenner09}. In \cite{WLGrenner09} the
present work is extended to quantum communication channels.

The remainder of this paper is arranged as follows: in
Section~\ref{sec:d0e} we introduce the quantity $D_0^\smooth (\cdot\|
\cdot)$; in Section~\ref{sec:converse} we state and prove the converse
theorem; in Section~\ref{sec:achiev} we state and prove the
achievability theorem; in Section~\ref{sec:asymp} we analyze
the bounds asymptotically for an arbitrary channel to study its
capacity and $\err$-capacity; finally, in Section~\ref{sec:compare} we
compare numerical results obtained using our new achievability bound
with some existing bounds.

\section{The Quantity $D_0^\smooth (\cdot\| \cdot)$}\label{sec:d0e}

In \cite{renyi61} R\'enyi defined entropies and
divergences of order $\alpha$ for every $\alpha>0$. We denote these
$H_\alpha(\cdot)$ and $D_\alpha(\cdot \| \cdot)$ respectively. They are
generalizations of Shannon's entropy $H(\cdot)$ and relative 
entropy $D(\cdot\|\cdot)$.

Letting $\alpha$ tend to zero in $D_\alpha(\cdot\|\cdot)$ yields the following
definition of $D_0(\cdot\| \cdot)$.

\begin{dfn}[R\'enyi's Divergence of Order $0$]
  	For $P$ and $Q$ which are two probability measures on $(\Omega,
        \set{F})$, $D_0(P\| Q)$ is defined as
	\begin{equation}\label{eqn:1}
		D_0(P\| Q)  =  -\log \int_{\supp(P)} \d Q,
	\end{equation}
        where we use the convention $\log 0 = -\infty$.\footnote{We
          remark that for 
          distributions defined on a finite alphabet, $\set{X}$, the
          equivalent of \eqref{eqn:1} is
          $D_0(P||Q)=-\log\sum_{x:P(x)>0}Q(x)$.}
\end{dfn}

We generalize $D_0(\cdot\| \cdot)$ to define $D_0^\smooth(\cdot\|
\cdot)$ as follows.

\begin{dfn}[Smooth $0$-Divergence]\label{dfn:d0e}
  	Let $P$ and $Q$ be two probability measures on $(\Omega,
        \set{F})$. For $\smooth>0$, $D_0^\smooth (P\| Q)$ is defined as
	\begin{equation}\label{eq:d0e}
		D_0^\smooth  (P\| Q) = \sup_{\substack{\Phi: \Omega
                    \to [0,1]\\\int_\Omega \Phi\d P \ge 1- \smooth}}
                \left\{ - \log \int_\Omega \Phi\d Q \right\}.
	\end{equation}
\end{dfn}

\emph{Remark:} To achieve the supremum in~\eqref{eq:d0e}, one should
choose $\Phi$ to be large (equal to $1$) for large $\frac{\d P}{\d Q}$
and vice versa.


\begin{lem}[Properties of $D_0^\smooth(\cdot\|\cdot)$]\label{lem:DPI}\ 
\begin{enumerate} 
  \item $D_0^\smooth(P\| Q)$ is monotonically nondecreasing in
    $\smooth$. 
  \item When $\delta=0$, the supremum in \eqref{eq:d0e} is achieved by
    choosing $\Phi$ to be $1$ on $\supp(P)$ and to be $0$
    elsewhere, which yields $D_0^0(P\| Q) = D_0(P\| Q)$.
  \item If $P$ has no point masses, then the supremum in \eqref{eq:d0e}
    is achieved by letting $\Phi$ take value in $\{0,1\}$ only and
    \begin{equation*}
      D_0^\smooth  (P\| Q) = \sup_{P': \frac{1}{2}\|P'-P\|_1 \le \smooth} D_0(P'
      \| Q).
    \end{equation*}
  \item (Data Processing Theorem) Let~$P$ and~$Q$ be probability
    measures on $(\Omega, \set{F})$, and 
  let~$W$ be a stochastic kernel from $(\Omega, \set{F})$ to
  $(\Omega', \set{F}')$. For all $\delta>0$, we have
  \begin{equation}\label{eq:DPI}
    D_0^\smooth (P\| Q) \ge D_0^\smooth (W\circ P\| W\circ Q),
  \end{equation}
  where $W\circ P$ denotes the probability distribution on
  $(\Omega', \set{F}')$ induced by $P$ and $W$ and similarly for
  $W\circ Q$.
  \end{enumerate}
\end{lem}

\begin{proof}
  The first three properties are immediate consequences of the
  definition and the remark. We therefore only prove~4).
  
  For any $\Phi':\Omega' \to [0,1]$ such that
  \begin{equation*}
    \int_{\Omega'} \Phi' \d (W\circ P) \ge 1-\smooth,
  \end{equation*}
  we choose $\Phi:\Omega\to \Reals$ to be
  \begin{equation*}
    \Phi(\omega) = \int_{\Omega'} \Phi'(\omega') W(\dd \omega'|\omega),
    \quad \omega \in \Omega.
  \end{equation*}
  then we have that $\Phi(\omega)\in [0,1]$ for all $\omega \in
  \Omega$. Further,
  \begin{IEEEeqnarray*}{rCl}
    \int_{\Omega}\Phi\d P & = &\int_{\Omega'} \Phi'\d (W\circ P) \ge
    1-\delta,\\
    \int_{\Omega}\Phi\d Q & = &\int_{\Omega'} \Phi'\d (W\circ Q).
  \end{IEEEeqnarray*}
  Thus we have
  \begin{IEEEeqnarray*}{rCl}
    \lefteqn{\sup_{\substack{\Phi: \Omega
                    \to [0,1]\\\int_\Omega \Phi\d P \ge 1- \smooth}}
                \left\{ - \log \int_\Omega \Phi\d Q \right\}}~~~~~~~~~~\\
               & \ge& \sup_{\substack{\Phi': \Omega'
                    \to [0,1]\\\int_{\Omega'} \Phi'\d (W\circ P) \ge 1- \smooth}}
                \left\{ - \log \int_{\Omega'} \Phi'\d (W\circ Q) \right\},
  \end{IEEEeqnarray*}
  which proves~4).
\end{proof}

A relation between $D_0^\smooth(P\| Q)$, $D(P\| Q)$ and the information
spectrum methods \cite{hanverdu93,han03} can be seen in the next
lemma. A slightly different quantum version of this theorem has been
proven in \cite{datta08}. We include a classical proof of it in the
Appendix.

\begin{lem}\label{lem:infrate}
	Let $P_n$ and $Q_n$ be probability measures on $\left(\Omega_n, \set{F}_n\right)$ for every $n\in\Naturals$. Then
	\begin{equation}\label{eq:infrate}
		\lim_{ \smooth  \downarrow 0} \varliminf_{n \to \infty} \frac{1}{n} D_0^{ \smooth } \left(P_n \| Q_n \right)= \left\{ P_n \right\}\textnormal{-}\varliminf_{n \to \infty} \frac{1}{n} \log \frac{\d P_n}{\d Q_n}.
	\end{equation}
	Here $\{P_n\}\textnormal{-}\varliminf$ means the $\liminf$ in
        probability with respect to the sequence of probability
        measures $\{P_n\}$, that is, for a real stochastic process $\{Z_n\}$,
	\begin{equation*}
		\left\{ P_n \right\}\textnormal{-}\varliminf_{n \to \infty} Z_n \triangleq \sup \left\{ a \in \Reals: \lim_{n\to \infty} P_n\left( \{ Z_n < a \} \right) = 0 \right\}.
	\end{equation*}
	In particular, let $P^{\times n}$ and $Q^{\times n}$ denote the product distributions of $P$ and $Q$ respectively on $\left(\Omega^{\otimes n}, \set{F}^{\otimes n}\right)$, then
	\begin{equation}\label{eq:iidrate}
		\lim_{ \smooth \downarrow 0} \varliminf_{n\to\infty} \frac{1}{n} D_0^{ \smooth }\left(P^{\times n} \| Q^{\times n} \right) = D(P\| Q).
	\end{equation}
\end{lem}

\begin{proof}
	See Appendix.
\end{proof}

\section{The Converse}\label{sec:converse}

We first state and prove a lemma.
\begin{lem}\label{lem:converse}
  Let $M$ be uniformly distributed over $\{1,\ldots,m\}$ and let
  $\hat{M}$ also take value in $\{1,\ldots,m\}$. If the
  probability that $\hat{M} \neq M$ is at most $\epsilon$, then 
  \begin{equation*}
    \log m \le D_0^\err \left(P_{M\hat{M}} \| P_M\times P_{\hat{M}}\right),
  \end{equation*}
  where $P_{M\hat{M}}$ denotes the joint distribution of $M$ and
  $\hat{M}$ while $P_M$ and $P_{\hat{M}}$ denote its marginals.
\end{lem}

\begin{proof}
  Let $\Phi$ be the indicator of the event $M=\hat{M}$, i.e.,
  \begin{equation*}
    \Phi(i,\hat{i}) \triangleq \begin{cases} 1,& i=\hat{i}\\0,
      & \textnormal{otherwise} \end{cases} \quad i,\hat{i}\in\{1,\ldots,m\}.
  \end{equation*}
  Because, by assumption, the probability that $M\neq \hat{M}$ is
  not larger than $\err$, we have
  \begin{equation*}
    \int_{\{1,\ldots,m\}^{\otimes 2}} \Phi \d P_{M\hat{M}} \ge 1-\err.
  \end{equation*}
  Thus, to prove the lemma, it suffices to show that
  \begin{equation}\label{eq:upperlem}
    \log m \le -\log \int_{\{1,\ldots,m\}^{\otimes 2}} \Phi \d (P_M \times
    P_{\hat{M}}).
  \end{equation}
  To justify this we write:
  \begin{IEEEeqnarray*}{rCl}
    \int_{\{1,\ldots,m\}^{\otimes 2}} \Phi \d (P_M \times
    P_{\hat{M}}) & = & \sum_{i=1}^m P_M\bigl(\{i\}\bigr)\cdot
    P_{\hat{M}} \bigl(\{i\}\bigr)\\
    & = & \sum_{i=1}^m \frac{1}{m}\cdot P_{\hat{M}}
    \bigl(\{i\}\bigr)\\
    & = & \frac{1}{m},
  \end{IEEEeqnarray*}
  from which it follows that \eqref{eq:upperlem} is satisfied with equality.
\end{proof}

\begin{thm}[Converse]\label{thm:converse}
  An $(m,\err)$-code satisfies
  \begin{equation}\label{eq:upper}
	\log m \le \sup_{P_X} D_0^\err \left(P_{XY} \| P_X \times P_Y \right),
  \end{equation}
  where $P_{XY}$ and $P_Y$ are probability distributions on
  $\set{X}\times\set{Y}$ and $\set{Y}$, respectively, induced by $P_X$
  and the channel law.
\end{thm}

\begin{proof}
  Choose $P_X$ to be the distribution induced by the message uniformly
  distributed 
  over $\{1,\ldots,m\}$, then 
  \begin{IEEEeqnarray*}{rCl}
    \log  m & \le & D_0^\err \left(P_{M\hat{M}} \| P_M\times
      P_{\hat{M}}\right)\\
    &\le & D_0^\err \left(P_{XY} \| P_X \times P_Y \right),
  \end{IEEEeqnarray*}
  where the first inequality follows by Lemma~\ref{lem:converse}; the
  second inequality by Lemma~\ref{lem:DPI} Part~4) and the fact that $M\markov
  X \markov Y \markov \hat{M}$ forms a Markov Chain. Theorem
  \ref{thm:converse} follows. 
\end{proof}

\section{Achievability}\label{sec:achiev}

\begin{thm}[Achievability]\label{thm:achievability}
  For any channel, any $\err>0$ and $\err'\in[0,\err)$
  there exists an $(m,\err)$-code satisfying
  \begin{equation}\label{eq:achievability}
    \log m \ge \sup_{P_X} D_0^{\err'}(P_{XY}\| P_X\times P_Y)
    -\log\frac{1}{\err-\err'},
  \end{equation}
  where $P_{XY}$ and $P_Y$ are induced by $P_X$ and the channel law.
\end{thm}

  The proof of Theorem~\ref{thm:achievability} can be thought of as a
  generalization of Shannon's 
  original achievability proof \cite{shannon48}. We use random coding
  as in 
  \cite{shannon48}; for the decoder, we generalize Shannon's
  typicality decoder to allow, instead of the ``indicator'' for the
  jointly 
  typical set, an arbitrary function on input-output pairs.

\begin{proof}
  For any distribution $P_X$ on $\set{X}$ and any $m\in\Integers^+$, we
  randomly generate a codebook of size $m$ such that the $m$
  codewords are independent and identically distributed according
  to $P_X$. We shall show that, for any 
  $\err'$, there exists a decoding rule associated with each
  codebook such that the average
  probability of a decoding error averaged over all such codebooks satisfies
  \begin{equation}\label{eq:lower1}
    \Pr(\textnormal{error}) \le (m-1) \cdot 2^{-D_0^{\err'}(P_{XY}\|
      P_X\times P_Y)} + \err'.
  \end{equation}
  Then there exists at
  least one codebook whose average probability of error is
  upper-bounded by the right hand side (RHS) of
  \eqref{eq:lower1}. That this codebook 
  satisfies \eqref{eq:achievability} follows by
  rearranging terms in \eqref{eq:lower1}.

  We shall next prove \eqref{eq:lower1}. For a given codebook and any
  $\Phi:\set{X}\times\set{Y}\to [0,1]$ which satisfies
  \begin{equation}\label{eq:lowerz}
    \int_{\set{X}\times\set{Y}} \Phi \d P_{XY} \ge 1-\err',
  \end{equation}
  we use the following
  random decoding rule:\footnote{It is well-known that, for the
    channel model considered in this paper, the average
    probability of error cannot be improved by allowing random
    decoding rules.} when $y$ is received, select some or none of the
  messages such that message~$j$ is selected with probability
  $\Phi(f(j),y)$ independently of the other messages. If only one message
  is selected, output this message; otherwise declare an error.

  To analyze the error probability, suppose $i$ was the transmitted
  message. The error event is the 
  union of $\set{E}_1$ and $\set{E}_2$, where $\set{E}_1$ denotes the
  event that some message other than~$i$ is selected; $\set{E}_2$ denotes the
  event that message $i$ is not selected.
  
  We first bound $\Pr(\set{E}_1)$ averaged over all codebooks. Fix
  $f(i)$ and $y$. The probability averaged over all codebooks of
  selecting a particular message other than~$i$ is given by
  \begin{equation*}
    \int_{\set{X}} \Phi(x,y)P_X(\dd x).
  \end{equation*}
  Since there are $(m-1)$ such messages, we can use the union bound to obtain
  \begin{equation}
    \E{\Pr(\set{E}_1 | f(i),y)} \le (m-1) \cdot \int_\set{X} \Phi(x,y)
    P_X(\dd x).\label{eq:lowera1} 
  \end{equation}
  Since the RHS of \eqref{eq:lowera1} does not
  depend on $f(i)$, we further have
  \begin{equation*}
    \E{\Pr(\set{E}_1 | y)} \le (m-1) \cdot \int_\set{X} \Phi(x,y) P_X(\dd x).
  \end{equation*}
  Averaging this inequality over $y$ gives
  \begin{IEEEeqnarray*}{rCl}
    \E{\Pr(\set{E}_1)} & \le & (m-1)  \int_{\set{Y}} \left(
      \int_\set{X} \Phi(x,y)P_X(\dd x) \right) P_Y (\dd y) \\ 
    & = & (m-1)  \int_{\set{X}\times\set{Y}} \Phi \d\left(P_X
      \times P_Y \right). \IEEEyesnumber \label{eq:lower21}
  \end{IEEEeqnarray*}
  On the other hand, the probability of $\set{E}_2$ averaged over all
  generated codebooks can be bounded as
  \begin{IEEEeqnarray}{rCl}
		\E{\Pr(\set{E}_2)} & = & \int_{\set{X}\times\set{Y}}
                (1-\Phi)\d P_{XY} \nonumber \\
		& \le &  \err' . \label{eq:lower22}
  \end{IEEEeqnarray}
  Combining \eqref{eq:lower21} and \eqref{eq:lower22} yields
  \begin{equation}\label{eq:lower3}
    \Pr(\textnormal{error}) \le (m-1) \int_{\set{X}\times\set{Y}} \Phi \d\left(P_X
      \times P_Y \right) + \err'.
  \end{equation}
  Finally, since \eqref{eq:lower3} holds for every
  $\Phi$ satisfying \eqref{eq:lowerz}, we
  establish  
  \eqref{eq:lower1} and thus conclude the proof of Theorem
  \ref{thm:achievability}.
\end{proof}

\section{Asymptotic Analysis}\label{sec:asymp}

In this section we use the new bounds to study the capacity of a
channel whose structure can be arbitrary. Such a channel is described by 
stochastic kernels from $\set{X}^n$ to $\set{Y}^n$ for all
$n\in\Integers^+$, where $\set{X}$ and
$\set{Y}$ are the input and output alphabets, respectively. An
\emph{$(n,M,\err)$-code} on a
channel consists of an encoder and a decoder such that a message of
size $M$ can be transmitted by mapping it to an element of
$\set{X}^n$ while 
the probability of error is no larger than $\err$. The capacity and the
optimistic capacity \cite{vembuverdusteinberg95} of a channel are
defined as follows. 

\begin{dfn}[Capacity and Optimistic Capacity]\label{dfn:capacity}
  The \emph{capacity} $C$ of a channel is the supremum over all $R$
  for which there exists a sequence of $(n,M_n,\err_n)$-codes such that
  \begin{equation}\label{eq:capacity}
    \frac{\log M_n}{n}\ge R,\quad n \in \Integers^+
  \end{equation}
  and
  \begin{equation*}
    \lim_{n\to\infty}\err_n=0.
  \end{equation*}
  The \emph{optimistic capacity} $\overline{C}$ of a channel is the
  supremum over all $R$ for which there exists a sequence of
  $(n,M_n,\err_n)$-codes such that \eqref{eq:capacity} holds and
  \begin{equation*}
    \varliminf_{n\to\infty}\err_n=0.
  \end{equation*}
\end{dfn}



Given Definition \ref{dfn:capacity}, the next theorem is an immediate
consequence of 
Theorems~\ref{thm:converse} and~\ref{thm:achievability}.

\begin{thm}[Capacity Formulas]
  Any channel satisfies
  \begin{IEEEeqnarray}{lCl}
    C & = & \lim_{\err\downarrow 0}
    \varliminf_{n\to\infty} \frac{1}{n} \sup_{P_{X^n}}
    D_0^\err\left(P_{X^nY^n}\| 
      P_{X^n}\times P_{Y^n}\right),\label{eq:verduhan}\\
    \overline{C} & = & \lim_{\err\downarrow 0}
    \varlimsup_{n\to\infty} \frac{1}{n} \sup_{P_{X^n}}
    D_0^\err\left(P_{X^nY^n}\| 
      P_{X^n}\times P_{Y^n}\right).\label{eq:chenalajaji1}
  \end{IEEEeqnarray}
\end{thm}

{\emph{Remark:}
According to Lemma~\ref{lem:infrate}, \eqref{eq:verduhan} is
equivalent to \cite[(1.4)]{verduhan94}. It can also be shown that
\eqref{eq:chenalajaji1} is equivalent to \cite[Theorem 4.4]{chenalajaji99}. }

We can also use Theorems~\ref{thm:converse}
and~\ref{thm:achievability} to study the $\err$-capacities which are
usually defined as follows (see, for example,
\cite{verduhan94,chenalajaji99}). 

\begin{dfn}[$\err$-Capacity and Optimistic $\err$-Capacity]
  The \emph{$\err$-capacity} $C_\err$ of a channel is the supremum over all $R$
  such that, for every large enough $n$, there exists an
  $(n,M_n,\err)$-code satisfying 
  \begin{equation*}
    \frac{\log M_n}{n}\ge R.
  \end{equation*}
  The \emph{optimistic $\err$-capacity} $\overline{C}_\err$ of a
  channel is the supremum 
  over all $R$ for which there exist $(n, M_{n},
  \err)$-codes for infinitely many $n$s satisfying
  \begin{equation*}
    \frac{\log M_n}{n}\ge R.
  \end{equation*}
\end{dfn}

The following bounds on the $\err$-capacity and optimistic
$\err$-capacity of a channel are immediate consequences of
Theorems~\ref{thm:converse} and~\ref{thm:achievability}. They can be
shown to be equivalent to those in 
\cite[Theorem 6]{verduhan94}, \cite[Theorem 7]{steinberg98} and
\cite[Theorem 4.3]{chenalajaji99}. As in those previous results, the
bounds for $C_\err$ ($\overline{C}_\err$) coincide except possibly at
the points of discontinuity of $C_\err$ ($\overline{C}_\err$).

\begin{thm}[Bounds on $\err$-Capacities]
  For any channel and any $\err \in (0,1)$, the $\err$-capacity of the channel
  satisfies
  \begin{IEEEeqnarray*}{rCl}
    C_\err & \le & 
    \varliminf_{n\to\infty} \frac{1}{n} \sup_{P_{X^n}}
    D_0^\err\left(P_{X^nY^n}\| 
      P_{X^n}\times P_{Y^n}\right),\\
    C_\err & \ge & \lim_{\err'\uparrow \err}
    \varliminf_{n\to\infty} \frac{1}{n} \sup_{P_{X^n}}
    D_0^{\err'}\left(P_{X^nY^n}\| 
      P_{X^n}\times P_{Y^n}\right);
  \end{IEEEeqnarray*}
  and the optimistic $\err$-capacity of the channel satisfies
  \begin{IEEEeqnarray*}{rCl}
    \overline{C}_\err & \le & 
    \varlimsup_{n\to\infty} \frac{1}{n} \sup_{P_{X^n}}
    D_0^\err\left(P_{X^nY^n}\| 
      P_{X^n}\times P_{Y^n}\right),\\
    \overline{C}_\err & \ge & \lim_{\err'\uparrow \err}
    \varlimsup_{n\to\infty} \frac{1}{n} \sup_{P_{X^n}} D_0^{\err'}\left(P_{X^nY^n}\|
      P_{X^n}\times P_{Y^n}\right).
  \end{IEEEeqnarray*}
\end{thm}



\section{Numerical Comparison with Existing Bounds for the
  BSC}\label{sec:compare}

In this section we compare the new achievability bound obtained in
this paper with the bounds by Gallager \cite{gallager65} and
Polyanskiy et al. \cite{polyanskiypoorverdu08}. We consider the
memoryless binary 
symmetric channel (BSC) with crossover probability $0.11$. Thus, for
$n$ channel uses, the
input and output alphabets are both $\{0,1\}^n$ and the
channel law is given by
\begin{equation*}
  P_{Y^n|X^n}(y^n|x^n) = 0.11^{|y^n-x^n|}0.89^{n-|y^n-x^n|},
\end{equation*}
where $|\cdot|$ denotes the Hamming weight of a binary vector. The
average block-error rate is chosen to be $10^{-3}$.

In the calculations of all three achievability bounds we choose
$P_{X^n}$ to be uniform on $\{0,1\}^n$. For comparison we include the
plot of the converse used 
in~\cite{polyanskiypoorverdu08}. Our new 
converse bound involves optimization over input distributions and is
thus difficult to compute. In fact, in this example it is less tight
compared to the 
one in~\cite{polyanskiypoorverdu08} since for the uniform
input distribution $D_0^{0.001}(P_{X^nY^n}\| P_{X^n}\times P_{Y^n})$
coincides with the latter. 

Comparison of the curves is shown
in Figure~\ref{fig:bsc}.
\begin{figure}[h]
  \centering
  \includegraphics[width=0.45\textwidth]{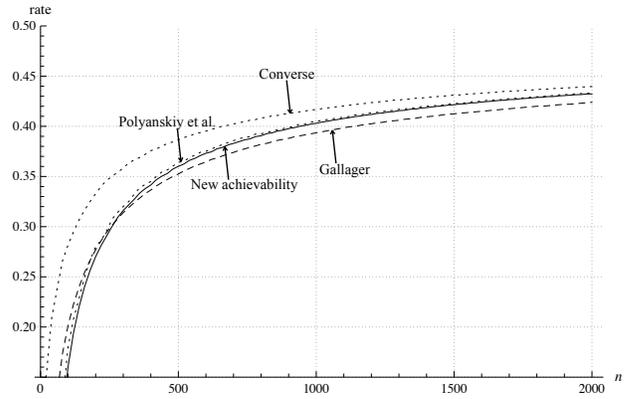}
  \caption{Comparison of the new achievability bound with
    Gallager~\cite{gallager65} and Polyanskiy et
    al.~\cite{polyanskiypoorverdu08} for the BSC with crossover
    probability $0.11$ and average block-error rate
    $10^{-3}$. The converse is the one used 
    in~\cite{polyanskiypoorverdu08}. \label{fig:bsc}}
\end{figure}
For the example we consider, the new achievability is always less tight
than the one in~\cite{polyanskiypoorverdu08}, though the difference is
small. It outperforms Gallager's bound for large block-lengths.

\appendix

In this appendix we prove Lemma~\ref{lem:infrate}. We first show that
\begin{equation}\label{eq:proofinf1}
	\lim_{ \smooth \downarrow 0} \varliminf_{n \to \infty} \frac{1}{n} D_0^{\smooth} \left(P_n \| Q_n \right) \ge \left\{ P_n \right\}\textnormal{-}\varliminf_{n \to \infty} \frac{1}{n} \log \frac{\d P_n}{\d Q_n}.
\end{equation}
To this end, consider any $a$ satisfying
\begin{equation}\label{eq:proofinf2}
	0 < a < \left\{ P_n \right\}\textnormal{-}\varliminf_{n \to \infty} \frac{1}{n} \log \frac{\d P_n}{\d Q_n}.
\end{equation}
Let $\set{A}_n(a) \in \set{F}_n$, $n\in\Naturals$, be the union of all measurable sets on which
\begin{equation}\label{eq:proofinf7}
	\frac{1}{n}\log \frac{\d P_n}{\d Q_n} \ge a.
\end{equation}
Let $\Phi_n:\Omega_n\to [0,1]$, $n\in\Naturals$, equal $1$ on $\set{A}_n(a)$ and equal $0$
elsewhere, then by \eqref{eq:proofinf2} we have
\begin{equation}\label{eq:proofinf8}
	\lim_{n\to\infty} \int_{\Omega_n}\Phi_n \d P_n =
        \lim_{n\to\infty} P_n\left(\set{A}_n(a)\right) = 1.
\end{equation}
Thus we have
\begin{IEEEeqnarray*}{rCl}
	\lefteqn{\lim_{\smooth\downarrow 0} \varliminf_{n\to\infty} \frac{1}{n}
        D_0^\smooth \left(P_n\|Q_n\right)}~~~~~~~~~\\
        &\ge& \varliminf_{n\to\infty} \left(- \frac{1}{n} \log
          \int_{\Omega_n} \Phi_n \d Q_n \right) \\
 	& \ge & \varliminf_{n\to\infty} \left(- \frac{1}{n} \log
          \int_{\Omega_n} \Phi_n \d P_n \cdot 2^{-na} \right) \\
	& = & \lim_{n\to\infty} \left(-\frac{1}{n} \log
          \left(2^{-na}\right) \right)\\ 
	& = & a, \IEEEyesnumber \label{eq:proofinf3}
\end{IEEEeqnarray*}
where the first inequality follows because, according to
\eqref{eq:proofinf8}, for any $\smooth>0$, $\int_{\Omega_n} \Phi_n \d
P_n = P_n\left(\set{A}_n(a)\right) \ge 1-\delta$ for
large enough $n$; the second inequality by \eqref{eq:proofinf7} and
the fact that $\Phi_n$ is zero outside $\set{A}_n(a)$; the
next equality by \eqref{eq:proofinf8}. Since \eqref{eq:proofinf3}
holds for every $a$ 
satisfying \eqref{eq:proofinf2}, we obtain \eqref{eq:proofinf1}. 

We next show the other direction, namely, we show that
\begin{equation}\label{eq:proofinf21}
	\lim_{\smooth \downarrow 0} \varliminf_{n \to \infty} \frac{1}{n} D_0^{\smooth} \left(P_n \| Q_n \right) \le \left\{ P_n \right\}\textnormal{-}\varliminf_{n \to \infty} \frac{1}{n} \log \frac{\d P_n}{\d Q_n}.
\end{equation}
To this end, consider any
\begin{equation}\label{eq:proofinf22}
	b > \{P_n\} -\varliminf_{n\to\infty} \frac{1}{n} \log \frac{\d P_n}{\d Q_n}.
\end{equation}
Let $\set{A}_n'(b)$, $n\in\Naturals$, be the union of all measurable sets on which
\begin{equation}\label{eq:proofinf23}
	\frac{1}{n}\log\frac{\d P_n}{\d Q_n} \le b.
\end{equation}
By \eqref{eq:proofinf22} we have that there exists some $c \in (0,1]$ such that
\begin{equation}\label{eq:proofinf24}
	\varlimsup_{n\to\infty} P_n\left(\set{A}_n'(b)\right) = c.
\end{equation}
For every $\smooth \in (0,c)$, consider any sequence of
$\Phi_n:\Omega_n\to [0,1]$ satisfying
\begin{equation}\label{eq:proofinf25}
	\int_{\Omega_n} \Phi_n \d P_n \ge 1-\smooth,\quad n \in \Naturals.
\end{equation}
Combining \eqref{eq:proofinf24} and \eqref{eq:proofinf25} yields
\begin{equation}\label{eq:proofinf26}
	\varlimsup_{n\to\infty} \int_{\set{A}_n'(b)} \Phi_n \d P_n \ge c-\smooth.
\end{equation}
On the other hand, from \eqref{eq:proofinf23} it follows that
\begin{equation}\label{eq:proofinf27}
	\int_{\set{A}_n'(b)} \Phi_n \d Q_n \ge \int_{\set{A}_n'(b)}
        \Phi_n \d P_n \cdot 2 ^{-nb}. 
\end{equation}
Combining \eqref{eq:proofinf26} and \eqref{eq:proofinf27} yields
\begin{equation*}
	\varliminf_{n\to\infty} \left( - \frac{1}{n} \log \int_{\set{A}_n'(b)}
        \Phi_n \d Q_n \right) \le b.  
\end{equation*}
Thus we obtain that for every $\smooth\in(0,c)$ and every sequence
$\Phi_n:\Omega_n\to [0,1]$ satisfying \eqref{eq:proofinf25}, 
\begin{IEEEeqnarray*}{rCl}
	\lefteqn{\varliminf_{n\to\infty} \left( -\frac{1}{n} \log
            \int_{\Omega_n} 
        \Phi_n \d Q_n \right)}~~~~~~~~~~~~~\\
        & \le & \varliminf_{n\to\infty} \left(
        -\frac{1}{n} 
        \log \int_{\set{A}_n'(b)} \Phi_n \d Q_n \right) \\  
	& \le & b.
\end{IEEEeqnarray*}
This implies that, for every $\smooth \in (0,c)$,
\begin{equation}\label{eq:proofinf29}
	\varliminf_{n\to\infty} \frac{1}{n}D_0^\smooth \left(P_n \| Q_n \right) \le b.
\end{equation}
Inequality \eqref{eq:proofinf29} still holds when we take the limit $\smooth\downarrow 0$. Since this is true for every $b$ satisfying \eqref{eq:proofinf22}, we establish \eqref{eq:proofinf21}.

Combining \eqref{eq:proofinf1} and \eqref{eq:proofinf21} proves \eqref{eq:infrate}.

Finally, \eqref{eq:iidrate} follows from \eqref{eq:infrate} because, by the law of large numbers,
\begin{equation}
	\frac{1}{n}\log \frac{\d (P^{\times n}) }{\d (Q^{\times n}) }
        \to \E{\log \frac{\d P}{\d Q}} = D(P\| Q) 
\end{equation}
as $n\to\infty$ $P^{\times n}$-almost surely. This completes the proof of
Lemma~\ref{lem:infrate}. 

\section*{Acknowledgment}
  RR acknowledges support from the Swiss National Science Foundation (grant No. 200021-119868).


\bibliographystyle{IEEEtran}           
\bibliography{/Volumes/Data/wang/Library/texmf/tex/bibtex/header_short,/Volumes/Data/wang/Library/texmf/tex/bibtex/bibliofile}

\end{document}